%% file: npsensingskin.tex
\documentclass[a4paper,english,11pt]{jnsao}
\usepackage[utf8]{inputenc}
\usepackage{csquotes}
\usepackage{babel}
\usepackage{tikz}
\usepackage{enumitem}
\usepackage{comment}
\usepackage{mathtools}
\usepackage{mdframed}
\usepackage{algorithmic}
\usepackage[ruled,section]{algorithm}
\usepackage{nicecleveref}
\usepackage[textsize=footnotesize,color=DarkOrange!40]{todonotes}
\usepackage{float}
\usepackage{lipsum}
\usepackage{grffile} 
\usepackage[section]{placeins}
\usepackage{numprint}

\npstyleenglish
\nprounddigits{3}
\newcommand{\dotp}[2]{\langle #1, #2 \rangle}
\mdfdefinestyle{examplebg}{
    backgroundcolor=black!7,%
    hidealllines=true,%
    innertopmargin=-.3em,%
    innerbottommargin=.7em,%
    innerleftmargin=.7em,%
    innerrightmargin=.7em,%
}
\title{Non-planar sensing skins for structural health monitoring based on electrical resistance tomography}

\theoremstyle{definition}

\surroundwithmdframed[style=examplebg]{example}
\makeatletter
\crefname{ALC@unique}{line}{lines}
\makeatother

\let\citep\cite

\input{symbols}

\input{commands}

\author{
    Jyrki Jauhiainen\thanks{Department of Applied Physics, University of Eastern Finland, Kuopio, Finland. \email{jyrki.jauhiainen@uef.fi}}
    \and
    Mohammad Pour-Ghaz\thanks{Department of Civil Construction and Environmental Engineering, North Carolina State University, Raleigh, NC 27695, USA. \email{mpourgh@ncsu.edu}}
    \and
    Tuomo Valkonen\thanks{ModeMat, Escuela Politécnica Nacional, Quito, Ecuador \emph{and} Department of Mathematics and Statistics, University of Helsinki, Finland. \email{tuomo.valkonen@iki.fi}}
    \and
    Aku Seppänen\thanks{Department of Applied Physics, University of Eastern Finland, Kuopio, Finland. \email{aku.seppanen@uef.fi}}
}
\shortauthor{Jauhianen, Pour-Ghaz, Valkonen, and Seppänen}
\shorttitle{Non-planar ERT sensing skins for SHM}
\manuscriptlicense{}
\manuscriptcopyright{}

\acknowledgements{This project has received funding from the European Union's Horizon 2020 research and innovation programme under grant agreement No 764810. The research was also funded by the Academy of Finland (Centre of Excellence of Inverse Modelling and Imaging, 2018-2025, project 303801).
T.~Valkonen has been supported by the Academy of Finland grants 314701 and 320022.
}

\begin{document}
\maketitle

\begin{abstract}
    Electrical resistance tomography (ERT) -based distributed surface sensing systems, or sensing skins, offer alternative sensing techniques for structural health monitoring, providing capabilities for distributed sensing of, for example, damage, strain and temperature. Currently, however, the computational techniques utilized for sensing skins are limited to planar surfaces. In this paper, to overcome this limitation, we generalize the ERT-based surface sensing to non-planar surfaces covering arbitrarily shaped three-dimensional structures; We construct a framework in which we reformulate the image reconstruction problem of ERT using techniques of Riemannian geometry, and solve the resulting problem numerically. We test this framework in series of numerical and experimental studies. The results demonstrate that the feasibility of the proposed formulation and the applicability of ERT-based sensing skins for non-planar geometries.
\end{abstract}
\section{Introduction}
\label{sec:intro}
    A component of SHM is a sensor network  consisting of variety of sensors utilizing a variety of techniques, that continuously monitors the condition of the infrastructure \citep{worden2004overview}. While the sensing techniques have advanced significantly over the past twenty years, utilization of SHM to real-life infrastructure  is still relatively rare. Many factors contribute to the slow adaptation of SHM for infrastructure, including the high cost of implementing and maintaining, as well as difficulty of the interpretation of measurements. The interpretation of the measurements is especially challenging when a large number of discrete sensors are used without the utilization of a model-based interpretation approach.  Distributed sensors and sensing systems may offer an alternative that at times can be more cost effective. Especially, distributed sensors that are model-based and provide direct visualization of the data can overcome many of the limitations of discrete sensors. An example of such system is an electrical resistance tomography (ERT) -based sensing skin \citep{hallaji2014new}.
    
    ERT based sensing skin is a distributed surface sensing system that uses a layer of electrically conductive material (such as colloidal metallic paint \citep{hallaji2014new,Hallaji2014skin} or carbon nanotube film \citep{loh2009carbon,loh2007nanotube}) which is applied to a surface of a structure. Also, a set of electrodes are placed on the surface, and based on multiple electric current/potential excitations and measurements from the electrodes, the spatially distributed electrical conductivity of the sensing layer is reconstructed. The surface coating material is designed so that the changes in its electrical conductivity give information on physical or chemical conditions of the underlying structure.

    One application of ERT-based sensing skins is damage detection: \cite{loh2009carbon,loh2007nanotube,hallaji2014new,Hallaji2014skin} Cracking of the structure surface breaks also the sensor layer, decreasing the conductivity of the layer material locally. The ERT reconstruction, that represents the electrical conductivity of the layer, thus reveals the crack pattern on the surface. ERT based sensing skins have also been developed for detection of pressure changes \cite{chossat2015soft}, strains \cite{loh2009carbon,tallman2014damage}, pH changes \cite{hou2007spatial}, chloride ions \cite{citation:EITSens3}, and temperature distributions \cite{rashetnia2017detection}.

    In the above cited papers, ERT-based sensing skins were applied to planar geometry only. In many applications, however, the structures of interest have a complex three-dimensional geometry, and the surface to be monitored is non-planar; examples of such target structures are pipelines, pumps and pressure vessels.

    In addition to SHM, ERT-based sensing systems have been applied to robotics, where the sensing skin is used for detecting and localizing touch via pressure sensing \cite{alirezaei2007highly, alirezaei2009tactile, yousef2011tactile}. In publication \cite{silvera2012interpretation}, an ERT-based touch sensor made of conductive fabric was wound around an artificial arm. The winding did not cause wrinkles to the fabric, but since the fabric was bent, the geometry was non-planar. The computational model used in the study, however, assumed a planar geometry. Although earlier studies have indicated that at least certain sensing skin materials are very sensitive to stretching and bending \cite{alirezaei2007highly, alirezaei2009tactile}, neglecting these effects by the use of planar approximation did not cause significant reconstruction artifacts in  \cite{silvera2012interpretation}. Nevertheless. it is not guaranteed that the planar approximation works with all materials, especially when aiming at quantitative imaging \cite{Hallaji2014skin}. Even more importantly, in many potential SHM applications, the planar approximation of the sensor is impossible, because of the nontrivial topology of the surface. This is the case for example with all the geometries considered in the numerical and experimental studies of this paper (\cref{fig:Drawing}).

    Another application, very similar to SHM with sensing skin is the use of ERT with self-sensing materials \citep{tallman2014damage, tallman2015damage}. Recently, ERT imaging was applied to self-sensing composite tubes for damage detection, and the structure was non-planar \citep{thomas2019damage}. In this case, the 3D structure of the target material was modeled as in other 3D ERT applications \citep{vauhkonen2004image, brown2003electrical, loke1996practical}. While in the self-sensing applications, the structures – and thus also sensors – are inherently three-dimensional, in sensing skin applications the thickness of a sensor is several orders of magnitudes lower than its other dimensions. Clearly, this type of sensor can be modeled as a surface in three-dimensional space, and a full three-dimensional model would be unnecessarily complicated, making the computations more complex and more prone to numerical errors.
    
    In this paper, we formulate the problem of imaging a thin, electrically conductive surface material -- sensing skin -- applied on an arbitrarily shaped three-dimensional object by modeling it as a two-dimensional surface in the three-dimensional space, or, mathematically as a manifold. The mathematical framework of the formulation is referred to as \emph{differential geometry}. Since the formulation and its mathematical proofs are very technical, their details are left to an extended, technical version of this paper, published with open access in arXiv [arxiv-viite]. The focus of the journal paper is in the numerical and experimental evaluation of this approach. In numerical and experimental studies, we evaluate the approach in cases of three non-planar geometries. In these studies, we consider two target applications; crack detection and imaging of diffusive processes (such as heat conduction on solid materials).

\section{Non-planar ERT imaging}
    \label{sec:npei}
    In ERT imaging, the conductivity of the target is reconstructed from the voltage and current data obtained through a set of electrodes placed on the surface of the target.     
    Typically, the target is treated as a three-dimensional or as a planar two-dimensional domain. However, in order to reconstruct the conductivity of an arbitrary shaped sensing skin, we consider the target as an arbitrary surface in three-dimensions. 

    In this section, we first write a model that describes the ERT measurements given the surface conductivity; this is referred to as the \emph{forward model} of ERT, and it is approximated numerically using the finite element method (FEM). The \emph{inverse problem} of ERT is to reconstruct the conductivity given the current/potential measurements. The inverse problem is ill-posed in the sense that the "conventional" solutions to this problem are non-unique and extremely intolerant to measurement noise and modelling errors. For this reason, the solutions of the inverse problem require \emph{a priori} information on the conductivity, or \emph{regularization} of the problem \cite{kaipio2006statistical}. In this paper, we formulate the inverse problem as a regularized least-squares problem, where the data fidelity term utilizes the FEM approximation of the forward model.

\subsection{Modeling of measurements}
\label{sec:method}
    Consider a measurement setup in which the measurement data is obtained by sequentially setting each electrode to a known potential, grounding others, and measuring the electric current caused by potential difference. We note that many of the existing ERT measurement systems operate the other way round -- using current excitations and potential measurements. However, for a such system, the formulation of both the forward and inverse problem are analogous with the formulation written in this section. The choice of using potential excitations and current measurements is made, because the commercial measurement device employed for the experiments (Section 4) uses this procedure.

    The above described measurement setting constitutes the following forward problem: solve the electric current $I^p_k(\sigma)$ through each electrode $k$, given the spatially distributed conductivity
    $\sigma(x)$ (where $x=(x_1,x_2,x_3)$ is the spatial variable) and a set of electric potentials $U^p_k$ corresponding to an excitation $p$. We model this relation using the complete electrode model (CEM) \citep{cheng1989electrode} which consist of a partial differential equation and a set of boundary conditions,
    \begin{subequations}\label{eq:CEM}
        \begin{alignat}{4}
            \label{eq:CEM1}
                \quad\quad\quad\quad &&\nabla \cdot (\sigma(x) \nabla u^p(x)) &=0  &&x\in \xset, \\
            \label{eq:CEM2}
                \quad\quad\quad\quad &&u^p(x) + \zeta_{k} \sigma \dotp{\nabla u^p(x)}{\hat n} &= U^p_{k}   &&x\in \partial \xset_{e_{k}} , \\
            \label{eq:CEM3}
                \quad\quad\quad\quad &&\int_{\partial \xset_{e_k}} \sigma \dotp{\nabla u^p(x)}{\hat n}\: d\tilde S &= -I^p_{k},\\ 
                \label{eq:CEM4} 
                \quad\quad\quad\quad &&\sigma\dotp{\nabla u^p(x)}{\hat n} &=0   &&x\in \partial \xset \setminus \bigcup_{k=1}^L \partial \xset_{e_k},
        \end{alignat}
    \end{subequations}
    where $M \subset \R^3$ is a surface with boundary $\partial M$, $\partial \xset_{e_{k}}$ is the part of the $\partial\xset$ representing the edge of the $k$'th electrode, $\zeta_k$ is contact resistance, $-\hat{n}$ is an inward unit normal of $\partial \xset$ (i.e. a vector tangent to $\xset$, pointing inwards), and $L$ is the number of electrodes. In addition, the currents $I^p_k$ are required to satisfy Kirchhoff's law
    $
        \sum_{k=1}^L I^p_{k} = 0
    $.
    We write $d\tilde S$ for the infinitesimal length elements of the one-dimensional boundary $\partial M_{e_k}$.

    By calling $M$ a surface, we mean that we can parametrize $x=(\phi^1(y^1,y^2), \phi^2(y^1,y^2), \phi^3(y^1,y^2))$ \emph{locally} for some $(y^1,y^2) \in U \subset \R^2$ and some $x \in V \subset M$. This means that the functions and differential operators in \eqref{eq:CEM} are two-dimensional, and can be formally defined through Riemannian geometry.  
    \input{Figures/Layout/Drawing1.tex}
    Formally, we equip \emph{manifold} M with a \emph{Riemannian metric} $g$. Metric $g$ defines a product on tangent vectors analogous to a dot product in vector spaces (see \cref{fig:Drawing}) and it consequently defines the divergence and the gradient operators; 
    \begin{equation}
        \label{eq:grad}
        \nabla \cdot f =  \frac{1}{\sqrt{\abs{g}}} \sum_{i=1}^2  \partial_i \left( \sqrt{\abs{g}} f^i \right) \;\; \text{and}\;\; \nabla \hat f =  \sum_{i=1}^2 \sum_{j=1}^2 g^{ij} \partial_j \hat f \partial_i,
    \end{equation}
    where $f: M \to \R^2$ (e.g. $f= \sigma \nabla u^p$), $\hat f: M \to \R$ (e.g. $\hat f = u^p$). The maps $\partial_i$ generalize directional derivatives to $M$; technically $\partial_i: F(M) \to F(M)$, where $F(M)$ is a collection of differentiable functions on $M$ and $i=1,2$, form a local basis for the tangent plane. In this basis, $|g|$ is the determinant of the matrix formed from the components of $g$. Furthermore, $d\tilde S$ 
    in \eqref{eq:CEM} is the Riemannian volume measure of a curve (length in $\R^3$) in $(\partial M,g_\omega)$. Since in practice, M is an embbed manifold, we define $g$ as the \emph{pullback} of the standard dot product in $\R^3$ to $M$ and $g_\omega$ as the pullback of $g$ to $\partial M$ \cite{lee2013smooth,kenig2014calderon}.
    
    \subsection{Variational form and numerical approximation of the forward model}
    We approximate \eqref{eq:CEM} with a Galerkin finite element method, as described in detail in the technical \cref{app:impdet}.
    Indeed, by introducing test function $(v,V)$, we can write \eqref{eq:CEM} in a variational form
    \begin{equation}
        \label{eq:intss}
        \begin{aligned}
          &\int_{M} \sigma \dotp{\nabla v}{ \nabla u^p}_g dS + \sum_k^L \tfrac{1}{\zeta_k}\int_{\partial M_{e_k}} u^pv - \sum_k^L \tfrac{1}{\zeta_k}\int_{\partial M_{e_k}} u^pV_k d\tilde S \\&- \sum_k^L \int_{\partial M_{e_k}} \sigma \dotp{\nabla u^p(x)}{\hat n}_{g_\omega}V_k d\tilde S 
          = \sum_k^L \tfrac{1}{\zeta_k} \int_{\partial M_{e_k}} U_k(v-V_k) d\tilde S.
        \end{aligned}
    \end{equation}
    We write $dS$ for the infinitesimal area elements of the two-dimensional surface $M$.
    Notation-wise, the variational form (2.3) is almost the same as the one written for the 3D electrical \cite{VossThesis2020}. However, the functions in \eqref{eq:intss} are defined on only the surface $M \subset \R^3$, the differential operators according \eqref{eq:grad}, and inner products are defined with respect to the Riemannian metrics $g$ and $g_\omega$.
    
    Furthermore, by approximating $u^p = \sum_j^N u^p_j v_j$ and $I^p_k$ as $I^p_k = (\sum_j^{L-1} \tilde I_j n_j)_k$, where $v_j$ is piecewise linear and $n_j \in R^{L}$ such that the first component of $n_j$ is always $1$ and the $j+1$ component is $-1$ and other indices are zero, \cref{eq:intss} admits the matrix form
    \begin{equation}\label{eq:feacem}
        \begin{bmatrix}
            D_1 & 0 \\
            D_2 & D_3
        \end{bmatrix}
        \begin{bmatrix}
            \bar u\\
            \bar I
        \end{bmatrix}= \begin{bmatrix}
            \tilde U_1 \\ \tilde U_2
        \end{bmatrix},
    \end{equation}
    where the matrices
    \[
        \begin{aligned}
        (D_1)_{i,j} &= \int_{M} \sigma \dotp{\nabla v_i}{ \nabla v_j}_g dS + \sum_k^L \tfrac{1}{\zeta_k}\int_{\partial M_{e_k}}  v_jv_i  d\tilde S,\\
        (D_2)_{i,j} &= \tfrac{1}{\zeta_{i+1}}\int_{\partial M_{e_{i+1}}} v_j d\tilde S -\tfrac{1}{\zeta_1}\int_{\partial M_{e_1}}  v_j d\tilde S,\\
        (D_3)_{i,j} &= 
        \begin{cases} 
                2, & i=j \\ 
                1, & \text{otherwise}
        \end{cases}
    \end{aligned}
    \]
    and the vectors $ (\bar{u})_{i} = u^p_i$, $ (\tilde U_1)_i = \sum_k^L \frac{U_k}{\zeta_k} \int_{\partial M_{e_k}} v_i d\tilde S$, $(\bar I)_i = \tilde I_i$, and
    $
        (\tilde U_2)_i = -\sum_k^L \tfrac{1}{\zeta_k} \int_{\partial M_{e_k}} U_k(n_i)_k d\tilde S.
    $
    \subsection{Inverse imaging problem}
    \label{sec:IIP}
    We can now concatenate the simulated measurements to form a vector $I(\sigma) = ( I_1^1(\sigma),  ...,  I_n^n(\sigma) )^T$. Further, we denote the vector containing the corresponding measured data by $I^M$.

    The typical approach to solve the inverse problem of EIT is to solve a conductivity that minimizes the sum of a so-called \emph{data term}, $\tfrac{1}{2}\norm{L(I(\sigma) - I^M)}^2$, and a \emph{regularization functional} $F(\sigma)$. In sensing skin applications, however, we may improve the reconstruction quality by utilizing measurements $I_{\mathrm{ref}}^M$, measured from an initial stage where the sensing skin is intact \cite{Hallaji2014skin}, to compute a homogeneous estimate $\sigma_{\mathrm{ref}}$ for the initial (background) conductivity of the sensing skin;
    \begin{equation}\label{eq:minJH}
    \sigma_\text{ref} \defeq \arg \min\limits_{\sigma \in \R^+} \tfrac{1}{2} \norm{I(\sigma) - I^M_\text{ref}}^2.
    \end{equation}
    Based on this estimate, we compute a discrepancy term $\epsilon\defeq I^M_\text{ref} - I(\sigma_\text{ref})$ which gives an approximation of the modeling error caused by neglecting the inhomogeneity of the background conductivity of the sensing skin. To compensate for the modeling error in the reconstruction of the conductivity $\sigma$ in the \emph{subsequent} stages, we add this term into the model $I(\sigma)$ \cite{Hallaji2014skin}, and reconstruct the conductivity $\sigma$ as a solution of a minimization problem
    \begin{equation}\label{eq:minJ}
        \hat \sigma \defeq \arg\min\limits_{\sigma \in V} \tfrac{1}{2}\norm{L(I(\sigma) - I^M + \epsilon)}^2 + F(\sigma),
    \end{equation}
    where $V = \left\lbrace f(x) \in H_N(M) \;|\; \sigma_\text{min} \le f(x) \le \sigma_\text{max},\; \forall x \in M \right\rbrace$, $H_N(M)$ is a finite dimensional function space on $M$, and $L$ is a matrix for which $L^T L$ is so-called \emph{data precision matrix}. The matrix $L^TL$ accounts for the magnitude of noise in the measurements. Furthermore, the lower constraint $\sigma_\text{min} > 0$ comes from the natural, physics-based limit for the positivity of the conductivity and the upper constraint $\sigma_\text{max}$ restricts the conductivity from above whenever the maximum conductivity is known. In cases where the maximum conductivity is unknown, we set $\sigma_\text{max}$ to an arbitrary large number. 
    
    Note that the regularization function $F(\sigma)$ in \eqref{eq:minJ} is chosen depending on the information that is available about the conductivity prior to the measurements.
    In the numerical and experimental cases of the following sections, we consider two choices of regularization functionals. We note, however, that the non-planar ERT scheme proposed in this paper is not restricted to any particular choices of regularization. Although the above modeling error correction method based on the discrepancy term $\epsilon$ is highly approximative, it has shown to be useful in several cases with real data \cite{Hallaji2014skin}, and is thus used also in this paper. A more advanced formulation of the inverse problem for detecting complex crack patterns in the presence of inhomogeneous background was proposed in \cite{smyl2018detection}. We note that, if needed, this computational method would also be directly applicable to the non-planar ERT model described above.
\section{Numerical simulation studies}
    \input{Figures/Layout/Geometries.tex}
    We evaluate the proposed ERT imaging scheme with numerical simulation studies using two non-planar geometries; one resembles a pipe segment (first column in \cref{fig:Geom}) and one resembles a pressure vessel (second column in \cref{fig:Geom}). The figures also illustrate the locations of the electrodes. We note that majority of them are \emph {internal electrodes}, in the sense that they are surrounded by the sensing skin. Such a setting is chosen in order to improve the sensitivity of ERT measurements; the use of internal electrodes improves the quality of ERT reconstructions from the case where all electrodes are in the perimeter of the sensing skin even in planar geometries \citep{rashetnia2018electrical} -- in non-planar imaging the effect is presumably even stronger. Furthermore, we consider two target applications; crack detection and imaging of diffusive processes (such as distributed temperature sensing \citep{rashetnia2017detection} or strain measurement \citep{tallman2015tactile,tallman2016damage}).

    Both geometries are used to study crack detection (Cases 1 and 2). In each geometry, we consider five stages of cracking. In the first stage, stage 0, the sensing skin is intact and the conductivity is homogeneous. Measurements simulated in this stage are used as the reference measurements $I_{\mathrm{ref}}^M$ and utilized for computing the homogeneous background estimate \eqref{eq:minJH}. In the subsequent stages, to simulate evolving crack pattern, we lower the conductivity at the locations that correspond to the cracks. The diffusive process imaging is studied in Case 3, where the geometry is same as in Case 1. Here, the conductivity distribution is spatially smooth, and it evolves in the diffusive manner, mimicking an application where the surface temperature distribution is monitored using a sensing skin.

    \subsection{Specification of geometries and simulation of data}
    The first column in \cref{fig:Geom} shows the pipe segment geometry. The radius of the pipe segment is \numprint{0.10} m and it consists of three \numprint{0.10} m long straight cylindrical sections connected by two curved sections that both turn 90 degrees to from an "S"-shaped geometry. The three straight sections each have eight symmetrically placed electrodes on them and the two curved sections both have four electrodes on their convex side. These electrodes are square-shaped with \numprint{0.01} m side length. 

    The second column in \cref{fig:Geom} shows the geometry of a pressure vessel. The diameter of the pressure vessel is 1 m and the length of the cylindrical middle section is \numprint{1.5} m. The radius of curvature for the spherical top section is \numprint{2.125} m. Furthermore, the chamber has 3 cylindrical extensions. One of the extensions is attached to the top section of the chamber. The radius of this extension is \numprint{0.30} m. The other two extensions are attached on the cylindrical section. The radius of the larger horizontal extension is \numprint{0.25} m and the radius of the smaller diagonal extension is \numprint{0.20} m. On each extension, 8 electrodes are placed radially. Furthermore, the cylindrical section of the chamber has four layers of radially placed electrodes. The topmost and bottommost layers have 14 electrodes each, and the two layers in between have 7 and 6 electrodes. 
    The total number of electrodes is 65. The inner electrodes on the main chamber are square-shaped with side length of \numprint{0.05} m. The other are rectangular with side lengths of \numprint{0.05} m and \numprint{0.025} m.

    The FE mesh that we use in the data simulation for the pipe segment geometry has 92578 nodes and 184057 elements, and the FE mesh for the pressure vessel has 491679 nodes and 980160 elements. In each simulation, we initially set the surface conductivity to $\sigma(x) = 1$ S and use it to generate the reference measurements (stage 0). Subsequently, we generate measurements from 4 stages of varying conductivities, each stage being a continuation of the previous one (stages 1-4). When simulating cracks (Cases 1 and 2), stages 1-4 consists of spatially narrow areas of low conductivity, $\sigma(x_\text{crack}) = 10^{-7}$ S (top rows in \cref{fig:Pipe,fig:Vessel}). When simulating the spatially smooth distribution (Case 3), the minimum conductivity is set to  0.89 S in a single point on the curved surface, and it gradually increases to background value 1 S as function of space. To mimic the diffusive process, the size of the area with lowered conductivity is increased between consecutive stages from 1 to 4 (top row in \cref{fig:PipeSmooth}).

\subsection{Image reconstruction}
    \label{ssec:ir}
    We reconstruct the conductivity by solving the minimization problem of \eqref{eq:minJ}. 
    In the crack detection problems in Cases 1 and 2, we utilize total variation (TV) regularization \cite{rudin1992nonlinear}
    \[
        F(\sigma) = TV(\sigma).
    \]
    TV regularization penalizes the magnitude of the spatial gradient of $\sigma$ in $L^1$ norm and is often suitable for cases where the conductivity features sharp edges on relatively homogeneous background. TV regularization is shown to be feasible in ERT based crack detection \cite{Hallaji2014skin}.
    
    In Case 3, we utilize Gaussian smoothness regularization
    \[
        F(\sigma) = \norm{R_\Gamma (\sigma - \sigma_\text{ref}) }^2,
    \]
    where $R_\Gamma$ is given by $R_\Gamma= \Gamma^{-1/2}$,  
    $
        \Gamma_{i,j} = ae^{- \frac{\norm{x^i-x^j}^2}{2b^2}}
    $ \cite{lipponen2013electrical}, $x^i, \;x^j \in \R^3$ are the locations of the nodes $i$ and $j$ in the FE mesh, $a = 100$ and $b = 0.075$. This is often a feasible choice of  regularization functional in cases of diffusive phenomena, because it promotes spatial smoothness of the conductivity distribution.
    
    In all the studies, the matrix $L$ is diagonal with $[L]_{i,i} = 1000$ and the minimum conductivity is $\sigma_\text{min} = 10^{-4}$ S. In addition, we compute a homogeneous estimate $\sigma_\text{ref}$ using the measurements $I^M_\text{ref}$ at the reference stage (stage 0). We use this estimate to compute the approximation error term $\epsilon = I^M_\text{ref} - I(\sigma_\text{ref})$ as described in Section 2.3. In Cases 1 and 2, we also use the homogeneous estimate as the maximum constraint $\sigma_\text{max} = \sigma_\text{ref}$, which encompasses the idea that the cracks can never increase the conductivity of the conductive layer \citep{Hallaji2014skin}. In Case 3, we set $\sigma_\mathrm{max} = \infty$, that is, the conductivity distribution is not constrained from above. The meshes used in the image reconstruction are sparser than those used when simulating the data. For example, the mesh for the pipe segment has 10358 nodes and 20389 elements while the mesh for the pressure vessel mesh has 20565 nodes and 40532 elements. 

    To solve the minimzation problem \eqref{eq:minJ}, we utilize the recently published iterative Relaxed Inexact Gauss-Newton (RIPGN) algorithm \cite{ripgn}. RIPGN is a Gauss-Newton variant; it linearizes the non-linear operator $I(\sigma)$ of \eqref{eq:minJ} at each iterate, finds an approximate solution to the associated proximal problem using primal dual proximal splitting (the algorithm of Chambolle and Pock \cite{chambolle2010first}), and interpolates between this solution and the one computed at the previous iteration step. After computing each iterate, we check the convergence of the algorithm by comparing the value of the objective function in \eqref{eq:minJ} at the current iterate to the value objective function at the previous iterates. Furthermore, we limit the maximum amount of computed iterations to 30. 
    
    The reason for applying the RIPGN method to optimization in this paper is that it was shown to shown to be very effective both in 3D and planar 2D ERT \citep{ripgn}. We note, however, that standard Gauss-Newton and Newton methods based on smoothing the minimum and maximum constraints and the TV functional \citep{Hallaji2014skin,gonzalez2017isotropic} could be utilized as well. All the code used in the study was written in Julia (1.3.1). Computations were done on AMD Ryzen 9 3950X CPU with 64 GB of RAM (DDR4, 3800 MHz, CL15). Parts of the RIPGN algorithm utilize CUDA code. CUDA code was run on Nvidia RTX 2080 Ti GPU. 
\subsection{Results and discussion}
    \subsubsection{Case 1: Crack detection in pipeline}
    \input{Figures/Layout/Pipe.tex}
        The results of Case 1 are illustrated in \cref{fig:Pipe}. The  top row shows the (true) simulated conductivity, and the reconstructed conductivity is depicted in the bottom row. Each column corresponds to a different cracking stage. 

        In the first stage (\cref{fig:Pipe}, column 1), a crack forms at the middle section of the pipe segment. The reconstruction captures the shape of this crack quite accurately and only a small artifact is visible near the crack. The conductivity value at the crack is $10^{-4}$ S, which equals to $\sigma_\text{min}$. %
        
        In the second stage (\cref{fig:Pipe}, column 2), two new cracks appear in the pipe segment, on the side opposite to the crack in stage 1. The reconstruction shows these cracks clearly: The locations and lengths of the cracks are somewhat correct. The orientation of the upper crack is slightly biased, but this bias is insignificant from practical point of view. 
        
        In the third stage (\cref{fig:Pipe}, column 3) the first crack (state 1) is lengthened upwards and further extended to two branches, forming a "Y"-shaped crack. The reconstructed surface conductivity traces the "Y"-shape of the crack well. The junction of the branches is slightly dislocated, but the size of the crack is again well recovered. In the final stage (\cref{fig:Pipe}, column 4) the two small cracks of stage 2 are inter-connected, forming a single crack extending from top to the mid section of the pipe segment. Again, the crack is well tracked by the ERT reconstruction, yet a couple of very small defects appear next to it. Note that the cracks in the reconstructed conductivity are thicker than the simulated ones since the inversion mesh is sparser.

    \subsubsection{Case 2: Crack detection in pressure vessel}
    \input{Figures/Layout/Vessel.tex}
        \cref{fig:Vessel} shows the simulated and reconstructed conductivity in each cracking stage in Case 2 where the geometry corresponds to a part of a pressure vessel.

        The reconstructions in Case 2 trace the evolution of the crack pattern well. In all stages of cracking, the reconstruction quality is similar to that in Case 1, although a few more deficiencies are present. This small reduction in quality compared to Case 1 is, however, expected. The surface area of the pressure vessel is thirteen times larger than the surface area of the pipe segment in Case 1 and the geometry is far more complex. Overall, the results of Case 2 further confirm the feasibility of the non-planar 2D ERT to crack detection applications.
        \subsubsection{Case 3: Imaging of diffusive phenomena on surface}
        \cref{fig:PipeSmooth} shows the true conductivity and the reconstruction on each stage in Case 3. 
        In first stage (\cref{fig:PipeSmooth}, column 1), a spatially smooth region of low conductivity appears at the middle section of the pipe segment. In the subsequent stages (\cref{fig:PipeSmooth}, columns 2-4), the surface area of this region increases and the value within the region decreases further. Each reconstruction reflects the corresponding stage clearly and the deficiencies in these reconstructions are apparent only at the last two stages. These deficiencies, however, look similar to what is observed in 3D and planar 2D ERT studies \citep{lipponen2013electrical}, and seem to be related to the type of regularization that is used. The simulation clearly demonstrates that ERT imaging of diffusive phenomena is achievable also in non-planar geometry.
        In first stage (the fist column in \cref{fig:PipeSmooth}), a spatially smooth region of low conductivity appears at the middle section of the pipe segment. In the subsequent stages (columns 2-3 in \cref{fig:PipeSmooth}), the surface area of this region increases and the value within the region decreases further. Each reconstruction reflects the corresponding stage clearly and the deficiencies in these reconstructions are apparent only at the last two stages. These deficiencies, however, look similar to what is observed in 3D and planar 2D ERT studies \cite{lipponen2013electrical}, and seem to be related to the type of regularization that is used. The simulation clearly demonstrates that ERT imaging of diffusive phenomena is achievable also in non-planar geometry.
        \input{Figures/Layout/PipeSmooth.tex}
\section{Experimental study}
\subsection{Experimental setup and image reconstruction}
    For the experimental validation of the non-planar sensing skin, we used a setup where the outer surface of a hollow plastic cube was covered with conductive paint. We refer to the experimental test case as Case 4. The paint was a 1:10 mixture of graphite powder (manufactured by Cretacolor, \url{www.cretacolor.com}) and black coating paint (RUBBERcomp, manufactured by Maston, \url{www.maston.fi}). Side length of the cube was \numprint{0.2} m and bottom of the cube was open (last column in \cref{fig:Geom}). Each side of the cube had eight electrodes. On the vertical sides, five of these electrodes were inner electrodes, and the remaining three were shared with the adjacent sides. On the top side, this configuration was four and four. In total, the number of electrodes was 32. The electrodes were square-shaped, and the side length of an electrode was \numprint{0.012} m. The electrodes shared by two cube sides were bent along the edges.

    We measured the reference data in the initial stage in which the sensing skin was intact. Subsequently, we simulated the cracking of the underlying structure by cutting the surface of the paint layer with a knife. We generated four different stages of cracking and carried out the ERT measurements corresponding to each of these stages. The same approach to "physically simulating" different stages of cracking has been used previously in cases on planar geometries, e.g, in \citep{Hallaji2014skin,citation:EITSens3}. Based on these studies, the quality of ERT reconstructions is similar in cases where a sensing skin is damaged with knife and where real crack patterns of the same complexity are monitored on the surface of a, e.g., a concrete beam.

    We measured the data with an ERT device manufactured by Rocsole Ltd. (\url{www.rocsole.com}). This ERT device samples the currents with $1$ MHz frequency, and computes the current amplitudes from the samples using discrete Fourier transform. The device outputs the amplitudes for the excitation potentials and for the measured electric currents. The device selects the amplitude for the excitation potentials automatically. Furthermore, we used the $39$ kHz excitation frequency, and to reduce the measurement noise, the current amplitudes that we used in the reconstructions were one-minute time averages. 

    Similarly to Cases 1 and 2, we use TV regularization to for the crack reconstructions (see \cref{ssec:ir}). Furthermore, the parameters are the same in the numerical cases. \cref{sfig:cgeo} shows the FE mesh used in the inversion. This mesh has 12278 nodes and 23965 elements.
\subsection{Results and discussion of the experimental study}
\input{Figures/Layout/Cube.tex}
    The top row in \cref{fig:Cube} shows a photo of the sensing skin at each stage and the bottom row shows the corresponding reconstructions (Case 4). We highlight the crack made at each stage with light teal color and the cracks made at the previous stages are darkened; the cracks are very thin (less than 1 mm in thickness) and would otherwise be indistinguishable from the background. 

    In the first stage (\cref{fig:Cube}, column 1), we created a diagonal crack on one the vertical sides of the cube.
    Reconstruction shows this crack accurately, although a small gap is visible in the reconstruction; the actual crack is fully connected. In the second stage (\cref{fig:Cube}, column 2), we extended the first crack so that it reaches the top side of the cube. The reconstruction shows the location and size of this crack quite accurately, although the curved extension of the crack is wider than the initial crack at stage 1.
    
    In the third stage (\cref{fig:Cube}, column 3), we created an additional crack on the adjacent side of the cube. This crack is clearly visible in the reconstruction. In the final stage (\cref{fig:Cube}, column 4), we extended the crack made on the third stage so that it reaches through the top side to the adjacent side. This extended crack is correctly located by ERT, although the reconstruction shows a blocky area in the corner of the cube. This reconstruction artifact is an expected one, since the electrodes are quite far from the cube corners, and therefore the ERT measurements are less sensitive to conductivity variations in these areas. Note also that the cracks in the reconstructed conductivities are thicker than the actual cracks made on the physical sensing skin. This is, again, partly caused by the sparsity of the finite element mesh, and partly a result of limited sensitivity of ERT to thickness of the cracks \citep{Hallaji2014skin}.

\section{Conclusions}
    One goal of the structural health monitoring research is to develop cost-effective sensor technologies. ERT based sensing skins have been proposed as a cost-effective distributed surface sensing systems for SHM. In the previous studies, the sensing skins sensors have been planar. To extend the usability of the ERT-based sensing skin to more complex structures it is necessary consider non-planar sensing skins and computational models.

    In this paper, we formulated the computational model for ERT in the case of non-planar surface sensing. We gave a brief outline of the numerical scheme to reconstruct the non-planar surface conductivity of the sensing skin. In this scheme, we modelled the relationship between the measured electric currents and the known electric potentials on a surface of an arbitrary object in 3D, and we used this model to formulate a minimization problem that yields the conductivity as the solution. Furthermore, we studied the feasibility of the scheme with three sets of numerical simulations and one set of experimental data.

    In the synthetic cases, we acquired highly accurate reconstructions, and we observed only minor artifacts in the reconstructed conductivity. These artifacts were similar to what has been observed in previous planar sensing skins studies. With the measurement data, the reconstruction quality was slightly worse than in synthetic cases but sufficient for most practical applications. Furthermore, we noted that the reconstructions from the measurement data could be improved, for example, by using a model for inhomogeneous background conductivity or by using a different electrode arrangement. 

    Overall, the reformulation of ERT imaging problem by using non-planar surface model proved to be viable; we did not observe any loss of reconstruction quality that could be related to non-planarity of the sensing skins. We conclude that with the proposed approach, ERT-based sensing skin is viable in monitoring complicated non-planar surfaces. In the future, non-planar sensing skins should allow monitoring of complex industrial structures such as those in aerospace, civil and mechanical engineering.



\appendix

\section{Implementation details}
\begin{figure}[!tp]%
    \centering%
    \includegraphics[width=0.33\textwidth]{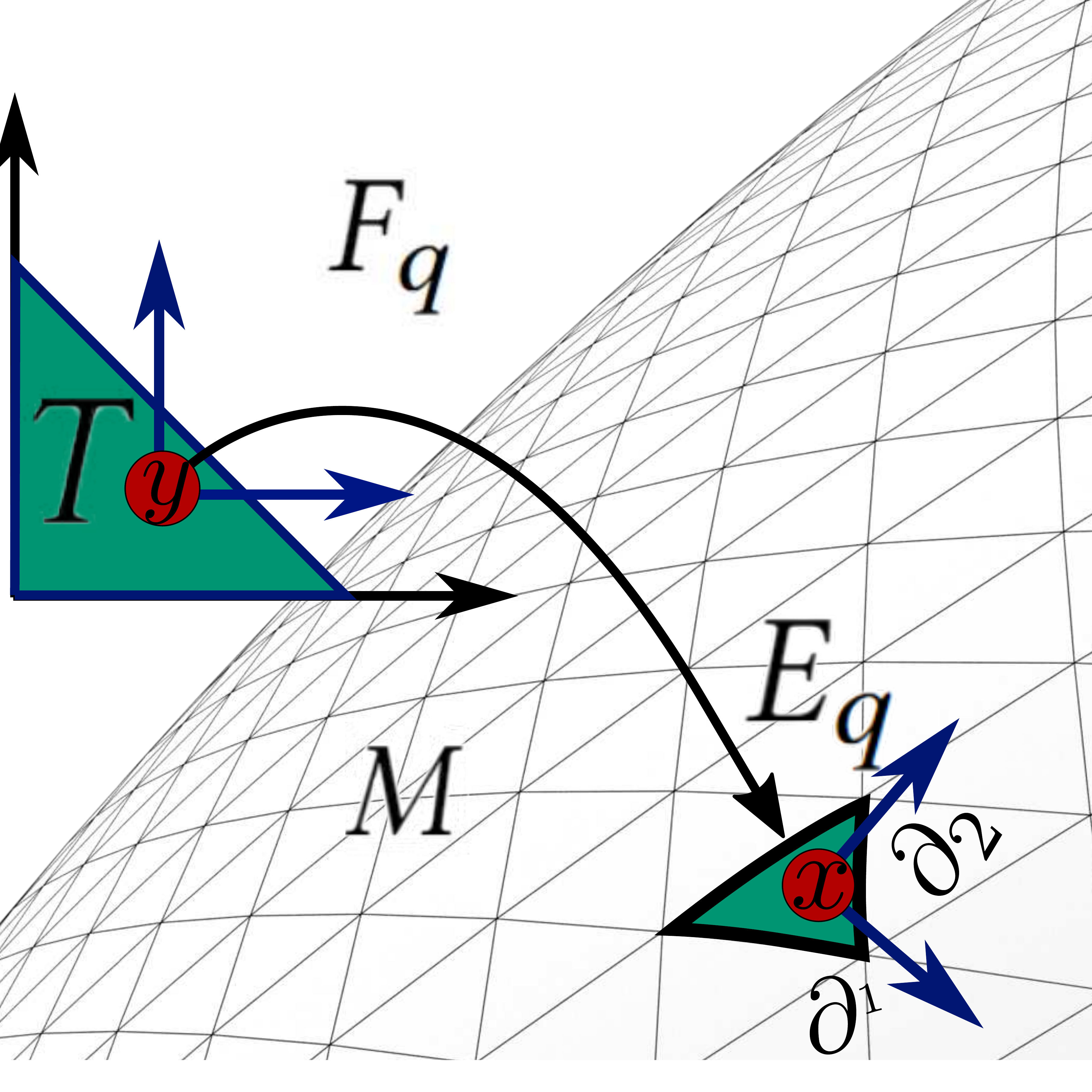}%
    \caption{Manifold $M$ is described locally on $T$ by diffeomorphism $F_q$. At point $x$, the tangents $\partial_1$ and $\partial_2$ set the basis for the tangent plane $T_xM$.}%
    \label{fig:Drawing2}%
\end{figure}%

\label{app:impdet}
Similarly to Euclidean spaces \cite{VossThesis2020}, we derive a finite element (FE) approximation for \eqref{eq:CEM}. Although \eqref{eq:CEM} looks identical to the Euclidean counterpart, the definitions of the operators in \eqref{eq:CEM} are more involved, containing calculations based on the Riemannian metric $g$. 

The FE approximation relies on the weak formulation of \eqref{eq:CEM}. The well-posedness of this weak formulation has been previously shown for $(u^p,U^p)$ (i.e. the potential measurement setup) \cite{somersalo1992existence} in Euclidean spaces, however for $(u^p,I^p)$ (i.e. the current measurement setup), no previous work exists; we will show the well-posedness of the weak formulation for $(u^p,I^p)$ in the manifold setting, which also extends to the Euclidean setting.

Initially, we take $g$ as an arbitrary metric on $M$. However, to see how to compute the FE approximation through integration in $\R^2$, we need to fix $g$. In this case, to properly account for the shape of $M$ in $\R^3$, we take $g$ as the metric induced on $M$ by the natural metric on $\R^3$ \cite{lee2018introduction,leeintroduction}. Namely, for tangents $w_1, w_2$ on the tangent plane $T_x M$ at a point $x$ (illustrated in \cref{fig:Drawing2}), it is defined by $g(w_1, w_2) \defeq \tilde g (d\phi(w_1), d\phi(w_2))$, where $\phi: M \to \R^3$ is the inclusion map $\phi(x) \defeq x$ and $\tilde g = (dx^1)^2 + (dx^2)^2 + (dx^3)^2$ is the Euclidean metric in $\R^3$.

The solutions $(u^p,I^p)$ of \eqref{eq:CEM} comprise a twice continuously differentiable function $u^p \in C^2 \defeq C^2(M)$ and a vector $I^p \in \R^L$ with components $I^p_k$, $k=1,\dots, L$. We denote $(u^p,I^p) \in C^2 \defeq C^2(M) \oplus \R^L$. We will show that the finite element approximation of \eqref{eq:CEM}, however, satisfies the weak formulation,
\begin{equation}\label{eq:bilin}
    B((u^p,I^p),(v,V)) = L((v,V)), \; \forall (v,V) \in H,
\end{equation}
where $B$ is bilinear and $L$ is linear. The space
\[
    H\defeq H^1(M) \oplus \R^L,
\]
where $H^1(M)$ is a Hilbert space of twice weakly differentiable functions.
We define it as the completion of $C^\infty(M)$ with respect to the norm $\norm{\cdot }_{H^1(M)}$ \cite[Chapter 10]{hebey2000non}.
It corresponds to the common space $H^1(\Omega)$ also used with planar CEM \cite{somersalo1992existence}.
The natural norm for this space is \cite{somersalo1992existence, hebey1996sobolev,hebey2000non} 
\begin{equation}\label{eq:Hnorm}
    \norm{(v,V)}^2_H=  \norm{v }_{H^1(M)}^2 + \norm{V}_{\R^L}^2,   
\end{equation}
where the inner products inducing the individual norms are 
\[
    \langle u, v\rangle_{H^1(M)} = \int_M uv dS + \int_M \dotp{\nabla u}{\nabla v}_g dS   
    \quad\text{and}\quad
    \langle U, V\rangle_{\R^L} = \sum_{k=1}^L U_k V_k.  
\]
In the following lemmas, we assume that the model \eqref{eq:CEM} has at least two electrodes, i.e. $L \ge 2$.

\begin{lemma}
    \label{lemmaweak}
    Suppose that $\zeta_k > 0$ is constant on $\partial M_{e_k}\;\forall k$, the part of $\partial M$ corresponding to electrode $k$. Then the PDE \eqref{eq:CEM} admits a weak formulation \eqref{eq:bilin}, where the bilinear operator $B: H \times H \to \R$ and the linear operator $L: H \to \R$ are given by
    \[        
        \begin{aligned}
        B((u^p,I^p),(v,V)) &= \int_{M} \sigma \dotp{\nabla v}{ \nabla u^p}_g dS + \sum_k^L \tfrac{1}{\zeta_k}\int_{\partial M_{e_k}} u^pv d\tilde S- \sum_k^L \tfrac{1}{\zeta_k}\int_{\partial M_{e_k}} u^pV_k d\tilde S + \sum_k^L I^p_kV_k 
    \end{aligned}
    \]
    and
    \[
        L(v,V) = \sum_k^L \tfrac{1}{\zeta_k} \int_{\partial M_{e_k}} U_k(v-V_k) d\tilde S.
    \]
\end{lemma}
\begin{proof}

    Suppose that $(u^p,I^p)$ solves \eqref{eq:CEM}. We need to show that it solves \eqref{eq:bilin}. So let $(v,V)\in H$ be arbitrary. We define $X\defeq\sigma \nabla u^p$. Applying $\int_{M} \cdot vdS$ to \eqref{eq:CEM1} we get
    \begin{equation}   
        \begin{aligned}
            -\int_{M} v \nabla \cdot (\sigma \nabla u^p) dS  &= -\int_{M} v \nabla \cdot X dS=0,
        \end{aligned}
    \end{equation}
    where $dS$ is the Riemannian volume corresponding to the metric $g$ on $M$.
    Denoting by $d\tilde S$ the Riemannian volume on $\partial M$, using the product rule, the divergence theorem on Riemannian manifolds \cite[Appendix A]{wang2012geometric}, and \eqref{eq:CEM2} to replace $\sigma \dotp{\nabla u^p}{\hat n}_{g_\omega}$, we obtain
    \begin{equation}\label{eq:cemd1}   
        \begin{aligned}[t]
            0=\int_{M}  v \nabla \cdot X dS  & = \int_{M} \dotp{\nabla v}{X}_g dS - \int_{\partial \xset} v \dotp{X}{\hat n}_{g_\omega} \: d\tilde S\\
            &= \int_{M} \dotp{\nabla v}{\sigma \nabla u^p}_g dS - \int_{\partial \xset} v \dotp{\sigma \nabla u^p}{\hat n}_{g_\omega}  \: d\tilde S\\
            &= \int_{M} \sigma \dotp{\nabla v}{ \nabla u^p}_g dS - \sum_k^L \int_{\partial M_{e_k}} v \sigma \dotp{\nabla u^p}{\hat n}_{g_\omega}  \: d\tilde S\\
            &= \int_{M} \sigma \dotp{\nabla v}{ \nabla u^p}_g dS - \sum_k^L \int_{\partial M_{e_k}} v (U^p_k-u^p)/\zeta_k \: d\tilde S.
        \end{aligned}
    \end{equation}
    The equations \eqref{eq:CEM2} and \eqref{eq:CEM3} both hold for each $k=1,\dots, L$ and define the vectors $U^p,I^p \in \R^L$. By multiplying each component $U^p_k$ of $U^p$ by $V_k/\zeta_k$, where  $V_k$ is a component of a test vector $V\in \R^L$, integrating over $\partial M_{e_k}$, and summing over $k=1,\ldots,L$, we get
    \begin{equation}\label{eq:cemd2}
        \begin{aligned}[t]
            \sum_k^L \int_{\partial M_{e_k}} u^p(x)V_k/\zeta_{k}\tilde{V} + \int_{\partial M_{e_k}}\sigma \dotp{\nabla u^p(x)}{\hat n}_{g_\omega}V_k\tilde{V} - \int_{\partial M_{e_k}} U^p_{k}V_k/\zeta_{k}d\tilde S = 0. 
        \end{aligned}
    \end{equation}
    Since $V_k$ is constant on $\partial M_{e_k}$,
    $       
        \int_{\partial M_{e_k}} \sigma \dotp{\nabla u^p(x)}{\hat n}_{g_\omega}V_k d\tilde S =  \int_{\partial M_{e_k}} \sigma \dotp{\nabla u^p(x)}{\hat n}_{g_\omega} d\tilde SV_k= -I^p_kV_k.
        $
    Subtracting \eqref{eq:cemd2} from \eqref{eq:cemd1} and plugging in $I^p_k$ gives
    \begin{equation*}
        \begin{aligned}[t]
            0-0&= \int_{M} \sigma \dotp{\nabla v}{ \nabla u^p}_g dS - \sum_k^L \int_{\partial M_{e_k}} v (U^p_k-u^p)/\zeta_k \: d\tilde S\\ 
            &-\left( \sum_k^L \int_{\partial M_{e_k}} u^p(x)V_k/\zeta_{k} - U^p_{k}V_k/\zeta_{k} d\tilde S  -  I^p_kV_k  \right)\\
            &= \int_{M} \sigma \dotp{\nabla v}{ \nabla u^p}_g dS + \sum_k^L \left(\int_{\partial M_{e_k}} u^p(v - V_k)/\zeta^k + U_k(V_k-v)/\zeta^k  d\tilde S + I^p_kV_k\right).
        \end{aligned}
    \end{equation*}
    Finally, since assume $\zeta_k$ is constant, by subtracting $\sum_k^L \int_{\partial M_{e_k}}  U_k(V_k-v)/\zeta^k d\tilde S$ we get \eqref{eq:bilin}.
\end{proof}

 The next lemma shows that the weak formulation \eqref{eq:bilin} is well-posed, meaning that the solution $(u^p,I^p)$ exists and is unique, and $B$ is continuous, leading eventually to the invertibility of the linear system of the FE approximation. For the simplicity, we assume that the boundary $\partial M$ of $M$ is $C^\infty$. However, the arguments that we use in the following proofs should extend to domains with boundaries of lesser smoothness.

Now, if we were solving for $(u^p,U^p)$ instead of $(u^p,I^p)$, we could follow the treatment in \cite{somersalo1992existence} by replacing relevant theorems on Sobolev spaces by their (compact Riemannian) manifold counterparts. However, no well-posedness proof for the weak formulation of $(u^p,I^p)$ exists. To prove the well-posedness for $(u^p,I^p)$, we show that the conditions of the Banach-Nes\v{c}a-Babu\v{s}ka theorem (BNB) hold for $B$ and that $B$ is continuous. The Euclidean case will follow as long as the domain for $u^p(x)$ is bounded.

\begin{lemma}\label{lemmawelp}
    Suppose that $0 < \sigma_m \le \sigma(x) \le \sigma_M < \infty$ is integrable on compact connected Riemannian manifold $(M,g)$ and with a $C^\infty$ boundary $\partial M$. Then \eqref{eq:bilin} is well-posed.
\end{lemma}
\begin{proof}
    For the proof, to avoid confusion between variables and not carry the index $p$, we write $(w, W)$ in place of $(u^p,I^p)$.
    According to BNB \cite[Theorem A.4 (Appendix)]{hesthaven2016certified} (See also \cite[Theorem 1]{saito2017notes}), since $H$ is a reflexive Banach space \cite[Proposition 2.1]{hebey2008sobolev} (Note that $H$ also a Hilbert space \cite[Proposition 2.1]{hebey2000non}), there exists a unique solution $(w,W) \in H$ to the problem \eqref{eq:bilin} if
    \begin{subequations}
        \begin{alignat}{2}
            \label{eq:bnb1} &\sup\limits_{{(v,V)}\in H} \frac{B((w,W),(v,V))}{\norm{(v,V)}_H} \ge \beta\norm{(w,W)}_H\text{ for some } \beta > 0\text{ and if} \\
            \label{eq:bnb2} &(\forall (w,W)\in H,B((w,W),(v,V)) =0)\Rightarrow ((v,V) = 0).
        \end{alignat}
    \end{subequations}
    First, however, similarly to \cite{somersalo1992existence}, we will show that 
    \[
        \norm{(v,V)}^2_*\defeq \int_{M} \dotp{\nabla v}{ \nabla v}_g dS + \norm{v}^2_{\partial M_e} + \norm{V}^2_{\R^L},
    \]
    where $\norm{v}^2_{\partial M_e} \defeq \sum_k^L \int_{\partial M_{e_k}} v^2 d \tilde V$, is a norm equivalent to \eqref{eq:Hnorm}, i.e. there exists constants $\lambda, \Lambda >0$ such that
    \begin{equation}
        \label{eq:manifold-normequiv}
        \Lambda\norm{(v,V)}_* \ge  \norm{(v,V)}_H \ge \lambda\norm{(v,V)}_*
        \quad \forall (v, V) \in H.
    \end{equation}

    To see that the first inequality of \eqref{eq:manifold-normequiv} holds, by the continuous embedding $H^{1/2}(\partial M) \subset L^2(\partial M)$ \cite[Definition 1.4, Chapter 4]{taylor2011partial} for some $C_1, C_2 > 0$,
    \[
        \norm{v}^2_{\partial M_e} \le \norm{v}^2_{L^2(\partial M)} \le C_1\norm{v}^2_{H^{1/2}(\partial M)} \le C_2\norm{v}^2_{H^{1}( M)}.
    \]
    Since $\dotp{\nabla v}{ \nabla v}_g \ge 0$, we thus obtain for some $\Lambda > 0$ that
    $
        \norm{(v,V)}^2_* \le \Lambda^2 (\norm{v}^2_{H^{1}( M)} + \norm{V}_{\R^L}^2)= \Lambda^2 \norm{v}^2_{H}.
    $
    
    To verify the second inequality of \eqref{eq:manifold-normequiv}, assume that the claim is not true. Then we can take a sequence $\lbrace (v^n,V^n) \rbrace_{n=1}^\infty\in H$, so that $\norm{(v^n,V^n)}_H = 1$ and $\norm{(v^n,V^n)}_* < 1/n$. Now, according to the compact embedding theorem on Sobolev spaces on manifolds \cite[Proposition 4.4, Chapter 4]{taylor2011partial}, $v^n $ contains a converging subsequence $ v^{n_i} \to v \in L^2(M)$, ${n_i} > n_{i-1}$, and $v \in L^2(M)$. Since $\frac{1}{{n_i}} >\norm{(v^{n_i},V^{n_i})}_* $, we have that 
    \begin{equation}\label{eq:triineq}
        \frac{1}{n_i^2} > \int_M \dotp{\nabla v^{n_i}}{\nabla v^{n_i}}_g dS,\quad  \frac{1}{n_i} > \norm{  v^{n_i}}_{\partial M_{e}}, \quad\text{and}\quad \frac{1}{{n_i}} > \norm{V^{n_i}}_{\R^L}.
    \end{equation}
        
    The first inequality implies that $v^{n_i}$ forms a converging sequence in $H^1(M)$ which satisfies $\int_M \dotp{\nabla v^{n_i}}{\nabla v^{n_i}}_g dS \to 0$. Applying Poincar\'{e} \cite[Theorem 2.10]{hebey2000non} and Hölder inequalities\footnote{Follows directly from Young's inequality.} shows for constants $c_1 \in \R$ and $ c_2 > 0$ that
    \[
        \norm{v^{n_i}-c_1}_{L^1(M)}^2 \le c_2\norm{\nabla v^{n_i} }_{L^1(M)}^2 \le c_2 \norm{\nabla v^{n_i} }_{L^2(M)}^2 \norm{1}_{L^2(M)}^2 \to 0,
    \]
    meaning that $v^{n_i}$ converges to the constant $c_1$ in $L^1(M)$. Further, since $v^{n_i} \to v \in L^2(M)$, using Hölder's inequality again shows that $\norm{v^{n_i}-v}_{L^1(M)} \le \norm{v^{n_i}-v}_{L^2(M)}\norm{1}_{L^2(M)} \to 0$, meaning that $v^{n_i} $ also converges to the same $v $ in $L^1(M)$, confirming that indeed $v=c_1$, i.e. $v$ is constant almost everywhere. Now, since $v$ is a.e. constant and $v \vert_{\partial M} \in L^2(\partial M)$, the second inequality in \eqref{eq:triineq} implies that $ v^{n_i} \to 0 $, i.e. $c_1=0$. The final inequality in \eqref{eq:triineq} implies that $ V^{n_i} \to 0$, i.e. $V=0$. Since $v^{n_i} \to 0$ and $V^{n_i} \to 0$, $\norm{( v^{n_i}, V^{n_i})}_H \to 0$, which is a contradiction, since $\norm{( v^{n_i}, V^{n_i})}_H = 1$.

    To see that \eqref{eq:bnb1} holds, start by denoting $s_k  = \int_{\partial M_{e_k}} 1 d\tilde S $, $a\defeq \max \left\lbrace\zeta_1^{-1},\zeta_2^{-1},\dots, s_1^2\zeta_1^{-1},\right.$ $\left. s_2^2\zeta_2^{-1}, \dots \right\rbrace$, and $c\defeq a\min \left\lbrace 1/a,\sigma_m, \zeta_1^{-1}, \zeta_2^{-1},\dots \right\rbrace$. If $(w,W)=(0,0)$, then \eqref{eq:bnb1} clearly holds. If $(w,W)\neq (0,0)$, then pick a function $(\hat v,\hat V)\in H$ that satisfies $\hat v = 2aw$ and $\hat V_k=W_k + \tfrac{1}{\zeta_k}\int_{\partial M_{e_k}}  w d \tilde V$. 

    Plugging $(\tilde v,\tilde V)$ into \eqref{eq:bilin} and simplifying gives \allowdisplaybreaks

    \begin{align*}
        &B((w,W),(\hat v,\hat V))  = 2a\int_{M} \sigma \dotp{\nabla w}{ \nabla w}_g dS  + \sum_k^L \tfrac{2a}{\zeta_k}\int_{\partial M_{e_k}} w^2 d\tilde S 
        - \sum_k^L \Bigl(\tfrac{1}{\zeta_k}\int_{\partial M_{e_k}} w d\tilde S\Bigr)^2
        +\sum_k^L W_k^2 \\
        &\ge 2a\int_{M} \sigma \dotp{\nabla w}{ \nabla w}_g dS 
        + \sum_k^L \tfrac{2a}{\zeta_k}\int_{\partial M_{e_k}} w^2 d\tilde S 
        - \sum_k^L \tfrac{a}{\zeta_k}\int_{\partial M_{e_k}} w^2 d\tilde S
        +\sum_k^L W_k^2 \\
        &\ge c\Bigl(\int_{M}  \dotp{\nabla w}{ \nabla w}_g dS+ \sum_k^L\int_{\partial M_{e_k}} w^2 d\tilde S+ \sum_k^L W_k^2\Bigr) =c \norm{(w,W)}_*^2
    \end{align*}
    Denoting $b\defeq 2\max \left\lbrace 1,2a^2, \vert e_1\vert^2\zeta_1^{-2},\vert e_2\vert^2\zeta_2^{-2},\dots \right\rbrace$,
    \[
        \begin{aligned}
        \norm{(\hat v,\hat V)}^2_* &
        \le 4a^2(\int_{M} \dotp{\nabla w}{ \nabla w}_g dS + \norm{w}^2_{\partial M_e}) + \sum_k 2( W_k^2 + (\tfrac{1}{\zeta_k}\int_{\partial M_{e_k}}  v dx)^2 )
        \\&\le 4a^2(\int_{M} \dotp{\nabla w}{ \nabla w}_g dS + \norm{w}^2_{\partial M_e})  + \sum_k ( 2W_k^2 + \tfrac{2s_k^2}{\zeta^2_k}\int_{\partial M_{e_k}}  v^2 dx )
        \\&\le 2b(\int_{M} \dotp{\nabla w}{ \nabla w}_g dS + \norm{w}^2_{\partial M_e} + {W}_{\R^L}^2) = 2b\norm{(w,W)}_*^2.
    \end{aligned}
    \] 
    Now since $\sup\limits_{{(v,V)}\in H} B((w,W),(v,V))/\norm{(v,V)}_H \ge  B((w,W),(\hat v,\hat V))/\norm{(\hat v,\hat V)}_H$, we have that \allowdisplaybreaks
    \begin{align*}
            &\sup\limits_{{(v,V)}\in H} B((w,W),(v,V))/\norm{(v,V)}_H  \ge \frac{c}{\Lambda} \norm{(w,W)}_*^2/\norm{(\hat v,\hat V)}_* \\
            & \ge \frac{c}{\Lambda}  \norm{(w,W)}_*^2/(\sqrt{2b} \norm{(w,W)}_*) =\frac{c}{ \sqrt{2b} \Lambda} \norm{(w,W)}_* \ge \frac{c}{ \sqrt{2b} \Lambda^2}\norm{(w,W)}_H. 
    \end{align*}
    To see that \eqref{eq:bnb2} holds, assume the contrary, i.e. there exists a $(v,V)\neq 0$ so that $B((w,W),(v,V))=0$ holds for all $(w,W)$. If $V=0$ choose $(w,W)=(v,0)$. If $V\neq 0$ choose $(w,W)=(0,V)$. Both scenarios show that $ B((w,W),(v,V))\neq 0$ with the chosen $(w,W)$, i.e. that $B((w,W),(v,V))=0$ does not hold for all $(w,W)$, which is a contradiction, meaning that the condition must hold.

    Finally, to see that $B$ is continuous, i.e. $B((w,W),(v,V)) \le \mathcal{C} \norm{(w,W)}_H\norm{(v,V)}_H$ for some $\mathcal{C} > 0$, observe that
    \[
        - \sum_k^L \tfrac{1}{\zeta_k}\int_{\partial M_{e_k}} wV_k d\tilde S \le \sum_k^L \left\rvert  \tfrac{1}{\zeta_k} \int_{\partial M_{e_k}} wV_k d\tilde S\right\rvert \le a \norm{w}_{\partial M_e} \norm{V}_{\R^L} \le a \norm{(w,W)}_{*}\norm{(v,V)}_{*}.
    \]
    Denoting $\tilde c\defeq \max \lbrace 1,\sigma_M,\zeta_1^{-1}, \zeta_2^{-1},\dots \rbrace$, clearly, 
    \[
        \begin{aligned}
             B((w,W),(v,V)) &\le  \tilde c \norm{(w,W)}_{*}\norm{(v,V)}_{*} - \sum_k^L \tfrac{1}{\zeta_k}\int_{\partial M_{e_k}} wV_k d\tilde S\\ 
             &\le (\tilde c + a) \norm{(w,W)}_{*}\norm{(v,V)}_{*} \le (\tilde c + a)\lambda^{-2} \norm{(w,W)}_{H}\norm{(v,V)}_{H}.
        \end{aligned}
    \]
    This finishes the proof.
\end{proof}
For the next lemma, we will replace $u^p$ and $I^p_k$ by their finite element approximations $u^p = \sum_j^N u^p_j v_j$ and $I^p = \sum_{j=1}^{L-1} (\tilde I_j n_j)$, where we allow $v_j$ to be an arbitrary FE basis function. For $I^p \in \R^L$, we fix basis vectors $n_j \in \R^L$ so that we can utilize Kirchhoff's law to eliminate one of the components: we choose vectors $n_j \in \R^{L}$ such that the components of $n_j$ are $(n_j)_1=1$, $(n_j)_{j+1}=-1$, and otherwise $(n_j)_{k}=0$.
This fixes the value of the $I^p_1$ so that $I_1^p = -\sum_{i=2}^{L} I^p_i $. Indeed, due to the Kirchhoff' law, we only have $L-1$ unknown currents. Note also that $n_j$ no longer appear in the lemma, since the value is easy to determine.

\begin{lemma}
    \label{lemmatrix}
    Replace $H$ by a finite dimensional subspace 
    \[
        H_N = \mathop{\mathrm{span}}\lbrace (v_1, 0),\dots,(v_N, 0), (0, n_1),\dots ,(0, n_{L-1})\rbrace .
    \]
    Then \eqref{eq:bilin} admits the presentation $D \theta = \bar U$, where $D \in \R^{(N+L-1) \times (N+L-1) }$ with
    \[
        D =
        \begin{bmatrix}
            D_1 & 0 \\
            D_2 & D_3,
        \end{bmatrix}
        \quad
        (D_1)_{i,j} = \int_{M} \sigma \dotp{\nabla v_i}{ \nabla v_j}_g dS + \sum_k^L \tfrac{1}{\zeta_k}\int_{\partial M_{e_k}}  v_jv_i  d\tilde S,
    \]
    \[
        (D_2)_{i,j} = \tfrac{1}{\zeta_{i+1}}\int_{\partial M_{e_{i+1}}} v_j d\tilde S -\tfrac{1}{\zeta_1}\int_{\partial M_{e_1}}  v_j d\tilde S,
    \quad
        (D_3)_{i,j} =
        \begin{cases} 
                2, & i=j \\ 
                1, & \text{otherwise}
        \end{cases}
    \] 
    and $\; \bar U \in \R^{N+L-1}$ 
    \begin{equation*}
        (\bar U)_i =
        \begin{cases}
            \sum_k^L \frac{U_k}{\zeta_k} \int_{\partial M_{e_k}} v_i d\tilde S, & i \le N \\
            \frac{U_{i+1}}{\zeta_{i+1}} \int_{\partial M_{e_{i+1}}}1d\tilde S - \frac{U_1}{\zeta_1} \int_{\partial M_{e_1}}1d\tilde S, & \text{otherwise}.
        \end{cases}
    \end{equation*}
    The vector $\theta = (\bar u^p, \bar I^p)$, where $(\bar u^p)_i = u^p_i$, $(\bar I^p)_i = \tilde I^p_i$ contains the coefficients of the finite element approximations for $u^p$ and $I^p_k$. Furthermore, the problem is well-posed, Galerkin orthogonality holds and for the exact solution $(\hat u^p, \hat I^p)$ and some constant $\mathcal{C} > 0$ we have
    \begin{equation}\label{eq:quasibest}
        \norm{(\hat u^p - u^p, \hat I^p - I^p)}_H \le \mathcal{C}  \inf\limits_{(v,V) \in H_N} \norm{(\hat u^p - v, \hat I^p - V^p)}_H.
    \end{equation}
\end{lemma}
\begin{proof}
    Since $B$ is continuous and since BNB holds for $B$, by applying the Cea's lemma \cite[Lemma 1]{lazarov2019inf}, we see that the problem is well-posed also in $H_N$ and the Galerkin orthogonality and the solution \eqref{eq:quasibest} hold for $(u^p, I^p) $. 
    
    Plugging in the expression for $u^p$ gives $N$ equations corresponding to each $(v,V)=(v_i,0)$:
    \begin{equation*}
        \begin{aligned}
            &\int_{M} \sigma \dotp{\nabla v_i}{ \nabla \sum_j^N u^p_j v_j}_g dS + \sum_k^L \tfrac{1}{\zeta_k}\int_{\partial M_{e_k}} \sum_j^N u^p_j v_jv_i  d\tilde S \\ &= \sum_j^N u^p_j \left( \int_{M} \sigma \dotp{\nabla v_i}{ \nabla v_j}_g dS + \sum_k^L \tfrac{1}{\zeta_k}\int_{\partial M_{e_k}}  v_jv_i  d\tilde S\right)
            = \sum_k^L \tfrac{1}{\zeta_k} \int_{\partial M_{e_k}} U_kv_i d\tilde S,
        \end{aligned}
    \end{equation*}
    which can be written with the matrix $D_1$ and vectors $ \bar{u}$ and $ (\bar U_1)_i = \sum_k^L \frac{U_k}{\zeta_k} \int_{\partial M_{e_k}} v_i d\tilde S$ as $D_1\bar u = \bar U_1$.
   Further, plugging in $I^p_k = \sum_{j=1}^{L-1} (\tilde I_j n_j)_k$ and $(v,V)=(0,n_i)$ gives additional $L-1$ equations:
    \begin{equation*}
        \begin{aligned}
        &-\sum_k^L \tfrac{1}{\zeta_k}\int_{\partial M_{e_k}} \sum_j^N u^p_j v_j n_i d\tilde S + \sum_k^L I^p_kV_k= -\sum_k^L \tfrac{1}{\zeta_k}\int_{\partial M_{e_k}} \sum_j^N u^p_j v_j n_i d\tilde S + \sum_k^L (\sum_j^{L-1} \tilde I_j n_j)_k (n_i)_k\\ &= \sum_j^N u^p_j \left( \tfrac{1}{\zeta_{i+1}}\int_{\partial M_{e_{i+1}}} v_j d\tilde S -\tfrac{1}{\zeta_1}\int_{\partial M_{e_1}}  v_j d\tilde S \right)+ \sum_j^{L-1} \tilde I_j \sum_k^L (n_j)_k (n_i)_k = -\sum_k^L \tfrac{1}{\zeta_k} \int_{\partial M_{e_k}} U_k(n_i)_k d\tilde S,
        \end{aligned}
    \end{equation*}
    which can be expressed with $D_2$, $D_2$, $\bar u$, $\bar U_1$, $\bar I$, and
    $
        (\bar U_2)_i = -\sum_k^L \tfrac{1}{\zeta_k} \int_{\partial M_{e_k}} U_k(n_i)_k d\tilde S = \frac{U_{i+1}}{\zeta_{i+1}} \int_{\partial M_{e_{i+1}}}1d\tilde S - \frac{U_1}{\zeta_1} \int_{\partial M_{e_1}}1d\tilde S
    $
    as $D_2 \bar u = \bar U_1$ and $D_3 \bar I=\bar U_2$. Finally, by combining the results we have  
    \begin{equation}\label{eq:feacem}
        \begin{bmatrix}
            D_1 & 0 \\
            D_2 & D_3
        \end{bmatrix}
        \begin{bmatrix}
            \bar u\\
            \bar I
        \end{bmatrix}= \begin{bmatrix}
            \bar U_1 \\ \bar U_2
        \end{bmatrix}.
    \end{equation}
    This finishes the proof.
\end{proof}

\subsection{Computing the finite element approximation}
In \cref{lemmatrix}, we derived the matrix presentation of the FE approximation. This matrix form, however, is abstract in the sense that the integrals are still presented in $M$ and $\partial M$. 

To proceed further, we consider specific maps to the manifold $M$. We define the elementary triangle by $T := \left\lbrace y \in \R^2 \; |\;  y^1, y^2 \ge 0,\; y^1 + y^2 \le 1 \right\rbrace$ and denote the boundary segments of this triangle by $\partial_1 T := \left\lbrace y \in T \; |\;  y^2 = 0 \right\rbrace$, $\partial_2 T := \left\lbrace y \in T \; |\;  \right. \left. y^1 = 0 \right\rbrace$, and $\partial_3 T := \left\lbrace y \in T \; |\;  y^1 = y^2\right\rbrace$. Further, we assume that there exists a triangulation of $M$, meaning that:
\begin{itemize}
    \item There are domains $E_q$ (in practice "geodesic triangles", compare \cref{fig:Drawing}), such that $M=E_1 \cup E_1...\cup E_K$.
    \item For each $q$ there exists an orientation preserving diffeomorphism $F_q(T) = E_q$ (i.e. $F_q:T \to E_q$ and $F_q^{-1} :E_q \to T$ are differentiable bijections and the Jacobian determinant of $F_q$ is positive). Intuitively, this means that $T$ can be morphed smoothly  to $E_q$ and $E_q$ back to $T$ through $F_q$ and $F_q^{-1}$, so that the right-hand (or left-hand) rule is preserved.
    \item The subdomains $E_q$ and $E_r$ for $q \neq r$ and the subdomains $F_q(\partial_\alpha T) \subset \partial M$ and $F_r(\partial_\beta T) \subset \partial M$ for $(q,\alpha) \neq (r,\beta)$ may only intersect at the boundaries. 
    \item For each electrode $e_k$ some collection $\mathcal{B}_k$ of $r$ and $\beta$ corresponds to the boundary segment of $e_k$, i.e. $\partial M_{e_k} = \bigcup\limits_{(r, \beta) \in \mathcal{B}_k} F_r(\partial_\beta T)$.
\end{itemize}

In addition, we need a presentation for the metric $g$ on $T$. Recall that $g$ has the specific form $g(w_1, w_2) = \tilde g (d\phi(w_1), d\phi(w_2))$. 
Clearly, this metric has a matrix presentation $g = J_q^TJ_q$ \cite[Chapter 5]{spink2017} on $T$, where $J_q$ is the Jacobian matrix of $F_q$. Similarly, if $\gamma_\alpha(t)$ is a curve that maps $[0,1]$ to one of the boundaries $\partial_\alpha T$, then $g_\omega(w) =  g (d\gamma_\alpha(w))$ has a presentation $g_\omega = {J_q}_{\gamma\alpha}^T{J_q}_{\gamma\alpha}$ in $[0,1]$, where ${J_q}_{\gamma\alpha}$ is the Jacobian of $F_q(\gamma_\alpha(t))$.

Now we derive the exact forms for the integrals of \eqref{lemmatrix} in terms of $T$ and the interval $[0,1]$.
The next lemma combined with \cref{lemmatrix} finally allows us to calculate the matrices and vectors of the FE system that determines $I(\sigma)$ in \eqref{eq:minJ}.
The RIPGN algorithm \cite{ripgn} utilizes these vectors and matrices to solve \eqref{eq:minJ}.

\def\II{\mathcal{I}}
\begin{lemma}\label{lemma:ints}
    Assume that there exists a triangulation $\lbrace E_q\rbrace $ of $M$. The integrals in \cref{lemmatrix} have the following presentations;
    \[
        \int_{M} \sigma \dotp{\nabla v_i}{ \nabla v_j}_g dS = \sum_q \int_{T} \sigma(F_q(y))   (\nabla_y v_i(F_q(y))  )^Tg^{-1}(\nabla_y v_j(F_q(y)) )  \sqrt{|g|} dy^1dy^2,
    \]
    \[
        \sum_k^L \tfrac{1}{\zeta_k}\int_{\partial M_{e_k}}  v_jv_i  d\tilde S = \sum_k^L \sum_{(q,\alpha) \in \II(i,j,k)}   \frac{U_k}{\zeta_k} \int_0^1 v_j(F_q(\gamma_\alpha(t))) v_i(F_q(\gamma_\alpha(t))) \sqrt{g_\omega} dt,
    \]
    \[
        \frac{U_k}{\zeta_k} \int_{\partial M_{e_k}} v_j d\tilde S = \sum_{(q,\alpha) \in \II(j,k)} \frac{U_k}{\zeta_k} \int_0^1 v_j(F_q(\gamma_\alpha(t))) \sqrt{g_\omega} dt,
    \] and
    \[
        \frac{U_{k}}{\zeta_k} \int_{\partial M_{e_{k}}}1d\tilde S = \sum_{(q,\alpha) \in \II(k)} \frac{U_k}{\zeta_k} \int_0^1 \sqrt{g_\omega} dt,
    \]
    where $g^{-1}$ is a matrix representing the coefficients of $g^{ij}$, $\nabla_y f$ is the gradient of $f$ with respect to the variable $y \in T \subset \R^2$, and 
    \begin{align*}
        \II(i,j,k) &\defeq \left\lbrace  (q,\alpha) \; |\; \text{$\partial F_q(\partial_\alpha T)$ is under an electrode $k$, } v_i(F_q(\partial_\alpha T)) \neq \lbrace 0 \rbrace \text{, and }  v_j(F_q(\partial_\alpha T)) \neq \lbrace0\rbrace.  \right\rbrace, \\
        \II(i,k) &\defeq \left\lbrace  (q,\alpha) \; |\; \text{$\partial F_q(\partial_\alpha T)$ is under an electrode $k$ and } v_i(F_q(\partial_\alpha T)) \neq \lbrace 0 \rbrace. \right\rbrace\text{, and} \\
        \II(k) &\defeq \left\lbrace  (q,\alpha) \; |\; \text{$\partial F_q(\partial_\alpha T)$ is under an electrode $k$.} \right\rbrace. 
    \end{align*}
\end{lemma}
\begin{proof}
    Since $T$ and $\partial T$ are compact, $F_q$ are diffeomophic, $\partial M_{e_k} = \bigcup_{(r, \beta) \in \mathcal{B}_k} F_r(\partial_\beta T)$ and $M=E_1 \cup E_1...\cup E_K$, where  $q \neq r$ and $(q,\alpha) \neq (r,\beta)$ only intersect at their boundaries, the conditions of \cite[Proposition 10.21]{leeintroduction} are met for $M$ and for each $\partial M_{e_k}$, and the integrals of \eqref{eq:bilin} defined in $ M$ and in $\partial M_{e_k}$ can be expressed as sums of Riemannian integrals in $\R$ or $\R^2$ over the sets $T$ and $\partial T$, i.e., for $f: M \to \R$ and $\tilde f: \partial M_{e_k} \to \R$, i.e.
    \[
        \int_{M} f dS = \sum_q \int_{T} F_q^*(fdS) = \sum_q \int_{T} (f\circ F_q) \sqrt{|g|} dy \quad\text{and}
    \]
    \[
        \int_{\partial M_{e_k}} \tilde f d\tilde S = \sum_{(r, \beta) \in \mathcal{B}_k} \int_{\partial_\beta T} F_r^*(\tilde fd \tilde V) = \sum_{(r, \beta) \in \mathcal{B}_k} \int_{\partial_\beta T} (\tilde f\circ F_r) \sqrt{|g_\omega|} d\tilde y
    \]   
    where $y\in T$ and $\tilde y \in \partial_\beta T$ \cite[Proposition 11.25, Proposition 15.31, and page 402]{lee2012introduction}.

    Since $g^{ij}$ are the indices of the inverse of the matrix representing $g$ \cite[Page 342]{lee2012introduction}, $ \nabla f = g^{-1}\nabla_y f$ in $T$. Furthermore, we denote $y \defeq (y^1,y^2) \in T$ and $dy\defeq dy^1dy^2$. Now 
\[
    \begin{aligned} 
        \int_{M} \sigma \dotp{\nabla v_i}{ \nabla v_j}_g dS 
        &= \sum_q \int_{T} \sigma(F_q (y)) \dotp{\nabla v_i(F_q (y))}{ \nabla v_j(F_q (y))} \sqrt{|g|} dy\\
        &= \sum_q \int_{T} \sigma(F_q(y))   (g^{-1} \nabla_y v_i(F_q(y))  )^Tg(g^{-1} \nabla_y v_j(F_q(y)) )  \sqrt{|g|} dy \\
        &= \sum_q \int_{T} \sigma(F_q(y))   (\nabla_y v_i(F_q(y))  )^Tg^{-1}(\nabla_y v_j(F_q(y)) )  \sqrt{|g|} dy , 
    \end{aligned}
\]
since $g$ (and $g^{-1}$) is symmetric.

Since $\partial M_{e_k} = \cup_{r, \beta} F_r(\partial_\beta T)$ for some $\mathcal{B}_k$, and since $F_r(\partial_\beta T)$ may only intersect at a single point, the boundary integrals in \cref{lemmatrix} can be mapped to the interval $[0,1]$ by composing the appropriate $F_r$ with one of the curves, $\gamma_1(t)=(t,0)$, $\gamma_2(t)=(0,t)$, or $\gamma_3(t)=(t,1-t)$, depending on which segments of  $\partial_\beta T$ constitute to $\partial M_{e_k}$ under $F_r$. On these boundaries, we may write $d \tilde V = \sqrt{|g_\omega|} dt$. Now, since the boundary integrals in \cref{lemmatrix} comprise only terms that correspond to an electrode $e_k$, and since these terms are zero if either a $v_i$ or a $v_j$ in the term is identically zero on $\partial M_{e_k}$, we are left with $\II$ as defined in the statement of the lemma. As an example
\[
    \begin{aligned} 
        &\sum_k^L \tfrac{1}{\zeta_k}\int_{\partial M_{e_k}}  v_jv_i  d\tilde S = \sum_k^L \sum_{(q, \alpha) \in \mathcal{B}_k} \int_{\partial_\alpha T}   \frac{U_k}{\zeta_k} v_j(F_q(\tilde y)) v_i(F_q(\tilde y)) \sqrt{g_\omega} d \tilde y\\
        &= \sum_k^L \sum_{(q, \alpha) \in \mathcal{B}_k} \int_0^1   \frac{U_k}{\zeta_k} v_j(F_q(\gamma_\alpha(t))) v_i(F_q(\gamma_\alpha(t))) \sqrt{g_\omega} d t.
    \end{aligned}   
\]
Now $\int_0^1   \frac{U_k}{\zeta_k} v_j(F_q(\gamma_\alpha(t))) v_i(F_q(\gamma_\alpha(t))) \sqrt{g_\omega} d t=0$ if either $v_i(F_q(\gamma(t))) \equiv 0 $ or $  v_j(F_q(\gamma(t))) \equiv 0$, meaning that we can replace $(q, \alpha) \in \mathcal{B}_k$ with ${(q,\alpha) \in \II(i,j,k)}$.

\end{proof}

\bibliographystyle{jnsao}
\bibliography{abbrevs,npsensingskin}


\end{document}

%% file: symbols.tex
\newcommand{\field}[1]{\mathbb{#1}}

\newcommand{\R}{\field{R}}

\newcommand{\norm}[1]{\|#1\|}

\newcommand{\abs}[1]{|#1|}

\newcommand{\Union}\bigcup
\newcommand{\Isect}\bigcap
\newcommand{\union}\cup
\newcommand{\isect}\cap
\newcommand{\bigunion}\bigcup
\newcommand{\bigisect}\bigcap

\newcommand{\defeq}{:=}

\makeatletter
\def \uminus@sym{\setbox0=\hbox{$\cup$}\rlap{\hbox 
        to\wd0{\hss\raise0.5ex\hbox{$\scriptscriptstyle{-}$}\hss}}\box0}
    \def \uminus    {\mathrel{\uminus@sym}}
\makeatother

\makeatletter
\def \weaktostar@sym{\setbox0=\hbox{$\rightharpoonup$}\rlap{\hbox 
        to\wd0{\hss\raise1ex\hbox{$\scriptscriptstyle{*\,}$}\hss}}\box0}
    \def \weaktostar    {\mathrel{\weaktostar@sym}}
\makeatother

\DeclareFontFamily{U}{mathx}{\hyphenchar\font45}
\DeclareFontShape{U}{mathx}{m}{n}{<-> mathx10}{}
\DeclareSymbolFont{mathx}{U}{mathx}{m}{n}
\DeclareMathAccent{\widebar}{0}{mathx}{"73}

\def\bar{\widebar}

\newcommand{\II}{\mathcal{I}}

%% file: commands.tex
\newcommand{\subfloatcolorbar}[4]{\hspace{#1}\vtop{\vskip#2\hbox{\includegraphics[width=#3,keepaspectratio,trim={0px 0px 0px 0px},clip]{#4}}}}%
\newcommand{\subcstageb}[9]{%
    \subfloat[#1]%
    {%
        \label{#2}
        \begin{minipage}{#3}%
            \includegraphics[width=1\linewidth, trim=#6 #7 #8 #9,clip]{#4}\\%
            \includegraphics[width=1\linewidth, trim=#6 #7 #8 #9,clip]{#5}%
        \end{minipage}%
    }%
}%
\newcommand{\subcstagec}[9]{%
    \subfloat[#1]%
    {%
        \label{#2}
        \begin{minipage}{#3}%
            \includegraphics[width=1\linewidth, trim=0 0 0 0,clip]{#4}\\%
            \includegraphics[width=1\linewidth, trim=#6 #7 #8 #9,clip]{#5}%
        \end{minipage}%
    }%
}%

\def\xset{M}


%% file: Figures/Layout/Drawing1.tex
\begin{figure}[!tp]%
    \centering%
    \includegraphics[width=0.25\textwidth]{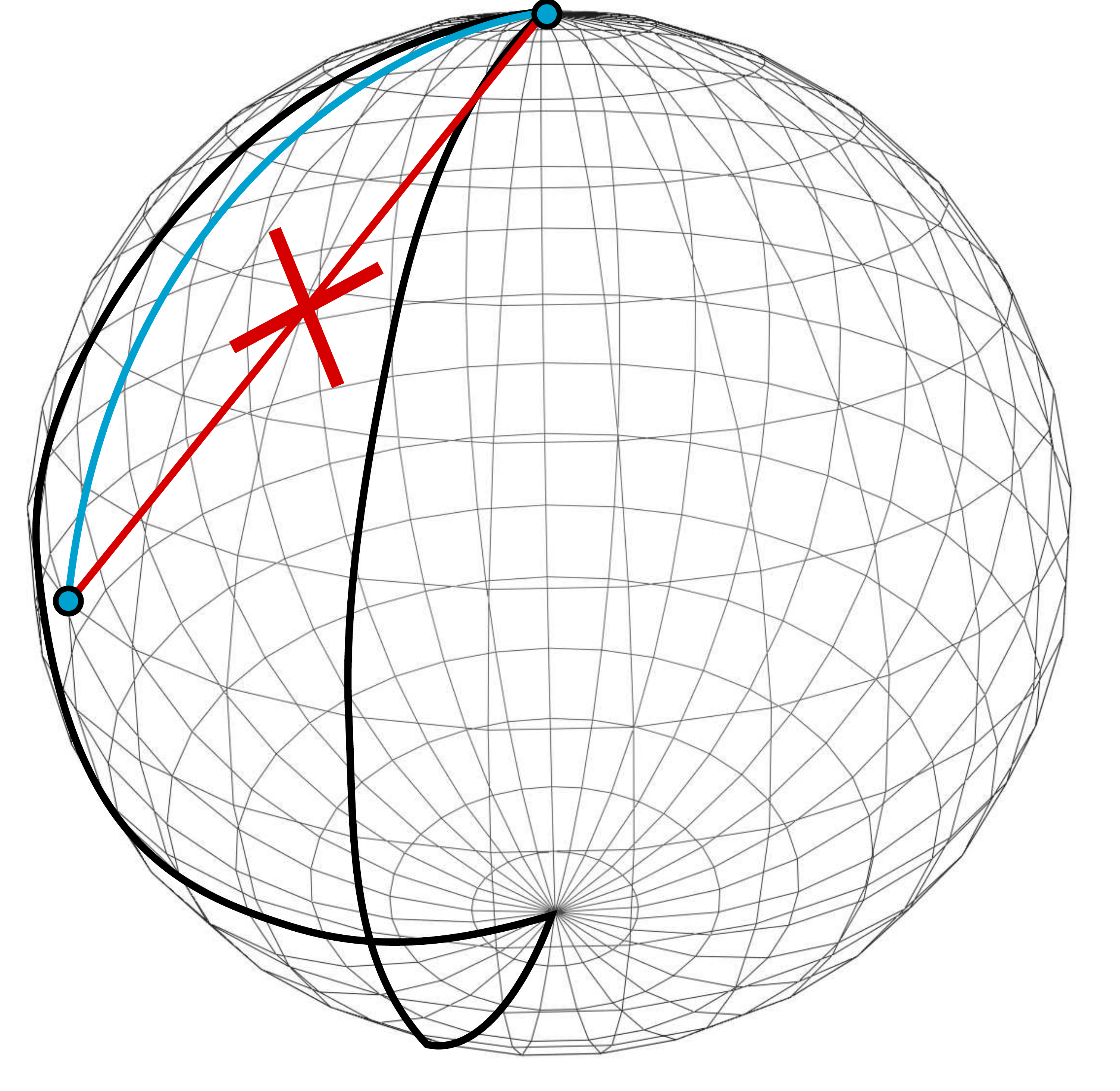}%
    \caption{An illustration of how the shortest path between two points in the non-planar two-dimensional model differs from the shortest path between these points in the three-dimensional model. Essentially the Riemannian metric determines how the distance is defined in the domain.}%
    \label{fig:Drawing}%
\end{figure}%

%% file: Figures/Layout/Geometries.tex
\captionsetup[subfloat]{position=top,labelformat=empty}%
\begin{figure}[t]
    \centerline{
\begin{minipage}{0.65\textwidth}%
    \centering%
        \subcstageb{}{sfig:psgeo}{0.333\linewidth}{"Figures/Pipe/PipeWF_w2".jpg}{"Figures/Pipe/PipeWF_w3".jpg}{80}{0}{100}{30}%
        \subcstageb{}{sfig:pvgeo}{0.333\linewidth}{"Figures/Chamber/ChamberWF_w1".jpg}{"Figures/Chamber/ChamberWF_w3".jpg}{80}{60}{160}{60}%
        \subcstageb{}{sfig:cgeo}{0.333\linewidth}{"Figures/Cube/CubeWF_w2".jpg}{"Figures/Cube/CubeWF_w1".jpg}{80}{0}{80}{0}%
    \end{minipage}}%
    \caption{Geometries of the sensing skins used in the numerical simulation studies (left and middle column) and in experimental study (right column). The surface triangulations correspond to the finite element meshes used in the respective image reconstructions. The square shaped non-triangulated patches of the surface represent the electrodes for the electrical measurements.}%
    \label{fig:Geom}%
\end{figure}%
\captionsetup[subfloat]{position=bottom}

%% file: Figures/Layout/Pipe.tex
\captionsetup[subfloat]{position=top,labelformat=empty}
\begin{figure*}[!t]%
\centering%
\begin{minipage}{0.9\textwidth}%
    \centering%
        \rotatebox{90}{\hspace{-6.2cm} Reconstruction \hspace{1.2cm} Simulation}
        \subcstageb{Stage 1.}{sfig:c1s1}{0.18\linewidth}{"Figures/Pipe/Pipe2_w1s".png}{"Figures/Pipe/Pipe2_w1".png}{150}{0}{150}{100}%
        \subcstageb{Stage 2.}{sfig:c1s2}{0.18\linewidth}{"Figures/Pipe/Pipe3_w2s".png}{"Figures/Pipe/Pipe3_w2".png}{150}{200}{200}{0}%
        \subcstageb{Stage 3.}{sfig:c1s3}{0.18\linewidth}{"Figures/Pipe/Pipe4_w1s".png}{"Figures/Pipe/Pipe4_w1".png}{150}{0}{150}{100}%
        \subcstageb{Stage 4.}{sfig:c1s4}{0.18\linewidth}{"Figures/Pipe/Pipe5_w2s".png}{"Figures/Pipe/Pipe5_w2".png}{150}{200}{200}{0}%
    \end{minipage}%
    \subfloatcolorbar{0.0cm}{-2.8cm}{0.030\textwidth}{"Figures/Pipe/Pipe2_cscale".png}%
    \caption{Case 1: True conductivity distribution of the sensing sensing skin (top row) and the ERT-based reconstructions of the conductivity (bottom row) corresponding to four stages of cracking.}%
    \label{fig:Pipe}%
\end{figure*}%
\captionsetup[subfloat]{position=bottom}

%% file: Figures/Layout/Vessel.tex
\captionsetup[subfloat]{position=top,labelformat=empty}
\begin{figure*}[!t]%
\centering%
\begin{minipage}{0.90\textwidth}%
    \centering%
        \rotatebox{90}{\hspace{-7.1cm} Reconstruction \hspace{1.9cm} Simulation}
        \subcstageb{Stage 1.}{sfig:c2s1}{0.18\linewidth}{"Figures/Chamber/Chamber2_w1s".png}{"Figures/Chamber/Chamber2_w1".png}{200}{50}{200}{200}%
        \subcstageb{Stage 2.}{sfig:c2s2}{0.18\linewidth}{"Figures/Chamber/Chamber3_w2s".png}{"Figures/Chamber/Chamber3_w2".png}{200}{50}{200}{200}%
        \subcstageb{Stage 3.}{sfig:c2s3}{0.18\linewidth}{"Figures/Chamber/Chamber4_w2s".png}{"Figures/Chamber/Chamber4_w2".png}{200}{50}{200}{200}%
        \subcstageb{Stage 4.}{sfig:c2s4}{0.18\linewidth}{"Figures/Chamber/Chamber5_w1s".png}{"Figures/Chamber/Chamber5_w1".png}{200}{50}{200}{200}%
    \end{minipage}%
    \subfloatcolorbar{0.0cm}{-2.8cm}{0.030\textwidth}{"Figures/Chamber/Chamber2_cscale".png}%
    \caption{Case 2: True conductivity distribution of the sensing sensing skin (top row) and the ERT-based reconstructions of the conductivity (bottom row) corresponding to four stages of cracking.}%
    \label{fig:Vessel}%
\end{figure*}%
\captionsetup[subfloat]{position=bottom}

%% file: Figures/Layout/PipeSmooth.tex
\captionsetup[subfloat]{position=top,labelformat=empty}
\begin{figure*}[t]%
\centering%
\begin{minipage}{0.90\textwidth}%
    \centering%
        \rotatebox{90}{\hspace{-7.1cm} Reconstruction \hspace{1.9cm} Simulation}
        \subcstageb{Stage 1.}{sfig:c4s1}{0.18\linewidth}{"Figures/PipeSmooth/PipeSmooth2_w1s".png}{"Figures/PipeSmooth/PipeSmooth2_w1".png}{150}{0}{150}{100}%
        \subcstageb{Stage 2.}{sfig:c4s2}{0.18\linewidth}{"Figures/PipeSmooth/PipeSmooth3_w1s".png}{"Figures/PipeSmooth/PipeSmooth3_w1".png}{150}{0}{150}{100}%
        \subcstageb{Stage 3.}{sfig:c4s3}{0.18\linewidth}{"Figures/PipeSmooth/PipeSmooth4_w1s".png}{"Figures/PipeSmooth/PipeSmooth4_w1".png}{150}{0}{150}{100}%
        \subcstageb{Stage 4.}{sfig:c4s4}{0.18\linewidth}{"Figures/PipeSmooth/PipeSmooth5_w1s".png}{"Figures/PipeSmooth/PipeSmooth5_w1".png}{150}{0}{150}{100}%
    \end{minipage}%
    \subfloatcolorbar{0.0cm}{-2.8cm}{0.030\textwidth}{"Figures/PipeSmooth/PipeSmooth5_cscale".png}%
    \caption{Case 3: True conductivity distribution of the sensing sensing skin (top row) and the ERT-based reconstructions of the conductivity (bottom row) corresponding to four stages of cracking.}%
    \label{fig:PipeSmooth}%
\end{figure*}%
\captionsetup[subfloat]{position=bottom}

%% file: Figures/Layout/Cube.tex
\captionsetup[subfloat]{position=top,labelformat=empty}
\begin{figure*}[!t]
\centering%
    \begin{minipage}{0.90\textwidth}%
    \centering%
    \rotatebox{90}{\hspace{-5.3cm} Reconstruction \hspace{1.0cm} Photo}
        \subcstagec{Stage 1.}{sfig:c3s1}{0.18\linewidth}{"Figures/Photos/Cube3".jpg}{"Figures/Cube/Cube3_w2".png}{100}{40}{100}{150}%
        \subcstagec{Stage 2.}{sfig:c3s2}{0.18\linewidth}{"Figures/Photos/Cube4".jpg}{"Figures/Cube/Cube4_w2".png}{100}{40}{100}{150}%
        \subcstagec{Stage 3.}{sfig:c3s3}{0.18\linewidth}{"Figures/Photos/Cube5".jpg}{"Figures/Cube/Cube5_w1".png}{100}{40}{100}{150}%
        \subcstagec{Stage 4.}{sfig:c3s4}{0.18\linewidth}{"Figures/Photos/Cube6".jpg}{"Figures/Cube/Cube6_w1".png}{100}{40}{100}{150}%
    \end{minipage}%
    \subfloatcolorbar{0.0cm}{-2.4cm}{0.026\textwidth}{"Figures/Cube/Cube2_cscale".png}%
    \caption{Case 4: Photographs of the sensing skin applied on the cubic object in the experimental study (top row) and the respective ERT-reconstructions (bottom row). The photos and reconstructions correspond to four stages of cracking; in the photographs, the cracks at each stage are highlighted with light teal color and the cracks if the previous stages are darkened.}%
    \label{fig:Cube}%
\end{figure*}%
\captionsetup[subfloat]{position=bottom,labelformat=simple}